\documentclass[11pt,a4paper]{article}
\usepackage{fullpage}
\usepackage{listings}
\usepackage{amsthm}
\usepackage{amsmath}
\usepackage{algorithm}
\usepackage{wrapfig} % to allow text to wrap around figures
\usepackage{url,hyperref}
\usepackage{graphicx}
\usepackage{array}
\usepackage{color}
\usepackage{float}
\usepackage{multirow}
\usepackage{smartref}
\usepackage[noadjust]{cite}

\newcommand{\ignore}[1]{}

\newcommand{\remove}[1]{}

\newcommand{\RedBlue}{{\sf RedBlue}}
\newcommand{\FRedBlue}{{\sf F-RedBlue}}
\newcommand{\SRedBlue}{{\sf S-RedBlue}}
\newcommand{\LSRedBlue}{{\sf LS-RedBlue}}
\newcommand{\BLSRedBlue}{{\sf BLS-RedBlue}}

\newcommand{\CX}{{\sf CX}}

\newcommand{\SIM}{{\sf Sim}} % Constant Time Simulation
\newcommand{\PSIM}{{\sf P-Sim}} % Constant Time Simulation
\newcommand{\LSIMOPT}{{\sf L-UC}} % Constant Time Simulation
\newcommand{\CAS}{{\tt CAS}}
\newcommand{\LC}{{\tt LL/SC}}

\newcommand{\LL}{{\tt LL}}
\newcommand{\SC}{{\tt SC}}
\newcommand{\VL}{{\tt VL}}
\newcommand{\FAD}{{\tt Add}}

\newcommand{\APPLYOP}{{\sc ApplyOp}}

\newcommand{\READ}{{\tt Read}}

\newcommand{\TRUE}{{\sl true}}
\newcommand{\FALSE}{{\sl false}}

\newcommand{\ATTEMPT}{{\tt Attempt}}

\newcommand{\LSIMAP}{{\sc LSimApplyOp}}

\newtheorem{theorem}{Theorem}[section]

\newtheorem{lemma}[theorem]{Lemma}

\title{An Efficient Universal Construction for Large Objects}

\date{}

\author{Panagiota Fatourou\\
	\small{Institute of Computer Science - Foundation for Research and Technology-Hellas (FORTH-ICS)}\\
    \small{\& Department of Computer Science, University of Crete, Greece}\\
    \textit{faturu@csd.uoc.gr}
    \and
    Nikolaos D. Kallimanis\\
    \small{Institute of Computer Science - Foundation for Research and Technology-Hellas (FORTH-ICS)}\\
    \textit{nkallima@ics.forth.gr}
    \and
    Eleni Kanellou\\
    \small{Institute of Computer Science - Foundation for Research and Technology-Hellas (FORTH-ICS)}\\
    \textit{kanelou@ics.forth.gr}}

\let\origthelstnumber\thelstnumber
\makeatletter
\newcommand*\Suppressnumber{%
  \lst@AddToHook{OnNewLine}{%
    \let\thelstnumber\relax%
     \advance\c@lstnumber-\@ne\relax%
    }%
}

\newcommand*\Reactivatenumber{%
  \lst@AddToHook{OnNewLine}{%
   \let\thelstnumber\origthelstnumber%
   \advance\c@lstnumber\@ne\relax}%
}

\begin{document}

\maketitle

\begin{abstract}
This paper presents \LSIMOPT, a universal construction that 
efficiently implements dynamic objects of large state in a wait-free manner. 
The step complexity of \LSIMOPT\ is $O(n+kw)$, 
where $n$ is the number of processes, 
$k$ is the interval contention (i.e., the maximum number of active processes 
during the execution interval of an operation), 
and $w$ is the worst-case time complexity 
to perform an operation on the sequential implementation of the simulated object. 
\LSIMOPT\ efficiently implements objects whose size can change dynamically.
It improves upon previous universal constructions 
either by efficiently handling objects whose state is large and can change dynamically,
or by achieving better step complexity.
\end{abstract}

\section{Introduction}
\label{sec:intro}
\subsection{Motivation and Contribution}
Multi-core processors are nowadays found in all computing devices. 
Concurrent data structures are frequently used as the means through
which processes communicate in multi-core contexts, 
thus it is important to have efficient and fault-tolerant implementations of them. 
A {\em universal construction}~\cite{H91,H93} provides an automatic mechanism 
to get a concurrent implementation of any data structure (or object) from its sequential implementation. 

In this paper, we present \LSIMOPT, an efficient, wait-free universal construction
that deals with dynamic objects whose state is large.
{\em Wait-freedom} ~\cite{H91} ensures that every process finishes the execution of
each operation it initiates within a finite number of steps. 
The step complexity of \LSIMOPT\ is $O(n+kw)$, 
where $n$ is the number of processes in the system, 
$k$ is the {\em interval contention}, i.e., the maximum number 
of processes that are active during the execution interval of an operation, 
and $w$ is the worst-case time complexity to perform an operation on the sequential
data structure. 
The step complexity of an algorithm is the maximum number of {\em shared memory accesses}
performed by a thread for applying any operation on the
simulated object in any execution.

A large number of the previously-presented universal constructions \cite{ADT95,AM95,CER10,FK09,FK11spaa,H91,H93}
work by copying the entire state of the simulated object locally, making the required updates 
on the local copy, and then trying to make the local copy shared by changing one (or a few) shared pointers
to point to it.
Copying the state of the object locally is however very inefficient when coping with large objects.
\LSIMOPT\ avoids copying the entire state of the simulated object locally; in contrast,
it applies the required changes directly on the shared state of the object.
For doing so, processes need to synchronize when applying the changes.
Previous universal constructions that apply changes directly to the shared data structure (e.g., \cite{CER10})
synchronize on the basis of each operation. However, this results in high synchronization cost. 
To reduce this cost, \LSIMOPT\ applies a wait-free analog of the combining technique~\cite{FK12,FK11spaa}:
each process simulates, in addition to its own operation, the operations of other active processes. 
So, in \LSIMOPT, processes have to pay the synchronization cost once for a batch of operations
and not for each distinct operation.

\SIM~\cite{FK11spaa, FK13} is a wait-free universal construction that implements the combining technique.
In \SIM, each process $p$ that wants to apply an operation,
first announces it in an $Announce$ array. 
Then, $p$ reads all other announced operations, makes a local copy of the shared state, 
applies all the operations it is aware of on this copy, and tries 
to update a shared variable to point to this local copy. 
\PSIM, the practical version of \SIM\ (presented also in~\cite{FK11spaa}) is
highly efficient for objects whose state is small. 
\LSIMOPT\ borrows some of the ideas presented in ~\cite{FK11spaa}.
Specifically, as \PSIM, \LSIMOPT\ uses an $Announce$ array 
in which processes announce their operations, and
employs bit vectors to figure out which processes have
active operations at each point in time. However, the bit vector mechanism 
of \LSIMOPT\ is more elaborated than  that of \PSIM, because the active processes 
have to agree on the set of operations that must be applied on the shared data structure
before they attempt to perform any changes.
In contrast to \SIM, \LSIMOPT\ avoids copying locally the object's state. 
This makes \LSIMOPT\ appropriate for simulating large objects.

\LSIMOPT\ also borrows some ideas from the universal construction presented in \cite{CER10} that copes with large objects.
As in the universal construction in~\cite{CER10}, in \LSIMOPT, each process uses a directory to store
copies of the shared variables (e.g.,  the shared nodes) it accesses while executing operations on the data structure.  
\LSIMOPT\ combines this idea with the idea of implementing a wait-free analog of the combining technique.
This way, \LSIMOPT\ achieves step complexity that is $O(n+kw)$. In scenarios of low contention, this bound can be 
much smaller than the $O(nw)$ achieved by the universal construction in~\cite{CER10}.
Moreover, the universal construction in~\cite{CER10} have processes synchronize on the basis
of every single operation, whereas in \LSIMOPT, processes synchronize once 
to execute a whole batch of operations.

\subsection{Related Work}

In~\cite{H91}, Herlihy studied how shared objects can be simulated,  
in a wait-free manner, using read-write registers and consensus objects. 
In the proposed universal construction, the simulated object is represented by a list of records.
Each record stores information about an operation $op$ 
(its type, its arguments, and its response) that has been performed on the simulated object. 
It also stores the state of the simulated object after all
operations inserted in the list up until $op$ (including it) have been applied on the implemented object 
in the order that they have been inserted in the list. 
To agree on which record will be inserted in the list next, 
each record additionally stores an $n$-consensus object.
To ensure wait-freedom, the algorithm also employs an
announce array of $n$ elements, where the $n$ threads running in the system 
announce their operations, and stores a (strictly increasing) sequence number in each record,
which illustrates the order in which this record was inserted in the list. 
Threads help the record of a thread $i$ to be inserted as the $j$-th record 
in the list when $i = j \mbox{ {\tt mod }} n$. 
The step complexity of the algorithm is $O(n^2)$.
The space overhead of the algorithm is $O(n^3)$ and each register contains
the entire state of the object and a sequence number growing infinitely large.
Herlihy revisited wait-free simulation of objects in~\cite{H93}, where it presented 
a universal construction which uses \LC\ and \CAS\ objects
and achieves step complexity $O(n+s)$, where $s$ is the total size of the simulated object. 
These algorithms~\cite{H91,H93} are inappropriate for large objects, 
as they work by copying the entire state of the object locally.

Afek, Dauber and Touitou presented in~\cite{ADT95} a universal construction that employs a 
tree structure to monitor which processes are {\em active}, i.e. which processes are performing 
an operation on the simulated object at a given time. This tree technique was combined %in ~\cite{H91}
with some of the techniques proposed in~\cite{H91,H93} in order to get a universal construction for simulating large objects,
which has step complexity $O(kw\log w)$.

Anderson and Moir presented in~\cite{AM99} a wait-free universal construction for simulating
large objects. In their algorithms, a contiguous array is used to represent
the state of the object. Specifically, the object state is stored in $B$ data blocks of size $S$ each. 
To restrict memory overhead, the algorithms operate under the following assumptions: 
each operation can modify at most $T$ blocks and each thread can help at most
$M \geq 2T$ other threads. The step complexity of the universal construction
in~\cite{AM99} is $O((n/\min\{k,M/T\})\,$ $(B +M S+nw))$.  %where 

In~\cite{FK09}, Fatourou and Kallimanis presented the family of \RedBlue\ adaptive 
universal constructions. The \FRedBlue\ algorithm achieves $O(min\{k, log n\})$ 
step complexity and uses $O(n^2+s$) \LC\ registers. 
However, \FRedBlue\ uses large registers and it is not able to simulate objects 
whose state is stored in more than one register. \SRedBlue\ uses small registers, but 
the application of an operation requires to copy the entire state of the 
simulated object and thus it is inefficient for large objects. % with large state.
\LSRedBlue\ and \BLSRedBlue\ improve the step complexity of the algorithms 
presented by Anderson and Moir in~\cite{AM99} for large objects.

In~\cite{correia2019waitfree}, Felber et al. present \CX, a wait-free universal 
construction, suitable for simulating large objects.
This universal construction keeps up to $2n$ instances of the object state.
In order to perform an update on the shared object, a process first appends
its request in a shared request queue and then attempts to obtain the lock of some of the 
object instances. We remark that each such object instance stores a pointer to a queue node.
Subsequently, the process uses this pointer to produce a valid copy of the object by performing 
all operations that were contained in the shared queue starting from the pointed node. 
Notice that \CX\ has space complexity $O(ns)$, where $n$ is the 
number of processes and $s$ is the total size of the simulated object.

\subsection{Roadmap}
The rest of this paper is organized as follows. Our model is discussed in 
Section \ref{sec:model}. \LSIMOPT\ is presented in Section~\ref{sec:description_lsimopt}.
Section~\ref{overview} provides an overview of the way the algorithm works and its pseudocode.
Section~\ref{detailed description} presents a detailed description of \LSIMOPT. 
A discussion of its complexity is provided in Section~\ref{step complexity}
and a sketch of proof for its correctness in Section~\ref{proof:lsimopt}.

\section{Model}
\label{sec:model}
We consider an  {\em asynchronous} system of $n$ {\em processes}, $p_1, \ldots,$
$p_n$, each of which may fail by {\em crashing}.
Threads communicate by accessing (shared) base objects. Each {\em base object} stores 
a value and supports some primitives in order to 
access its state.
An {\em \LC} object supports the atomic primitives
\LL\ and \SC. \LL($O$) returns the value that is stored into $O$. 
The execution of \SC$(O,v)$ by a thread $p_i$, $1 \leq i \leq n$,  
must follow the execution of \LL($O$) by $p$, and changes the contents of $O$ to $v$
if $O$ has not changed since the execution of $p$'s latest \LL\ on $O$.
If \SC$(O,v)$ changes the value of $O$ to $v$, \TRUE\ is returned
and we say that the \SC\ is successful; otherwise, the value of $O$ does 
not change, \FALSE\ is returned and we say that the \SC\ is not successful
or it is failed.\LSIMOPT\ is presented using \LC\ objects (as is the case for \SIM~\cite{FK11spaa, FK13}).
However,  
in a practical version of it, 
\LSIMOPT\ will be implemented using \CAS\ objects (as is the case for \PSIM~\cite{FK11spaa, FK13}). 
A \CAS\ object $O$ supports in addition to $\READ(O)$, the primitive
\CAS($O, u, v$) which stores $v$ to $O$ if the current value of $O$ is equal
to $u$ and returns \TRUE; otherwise the contents of $O$ remain unchanged and
\FALSE\ is returned.

A {\em universal} construction can be used to implement any shared object. 
A  universal construction supports the \APPLYOP($req$, $i$) operation,
which applies the operation (or {\em request}) $req$ to the simulated object and 
returns the return value of $req$ to the calling thread $p_i$.
In this paper, the concepts of an operation and a request 
have the same meaning and are used interchangeably.
A {\em universal construction} 
provides a routine, for each process, to implement \APPLYOP.

An object $O$ is {\em linearizable}, if in every execution $\alpha$,
it is possible to assign to each completed operation $op$ (and to some
of the uncompleted operations), 
a point $*_{op}$, called the {\em linearization point} of $op$, 
such that: $*_{op}$ follows the invocation and precedes the response of $op$, 
and the response returned by $op$ is the same as the response $op$
would return if all operations in $\alpha$ were executed sequentially in the
order imposed by the linearization points.

A {\em configuration} is a vector that contains the values of the base 
objects and the states of the processes, and describes the system at some point in time. 
At the {\em initial configuration}, processes are in their initial state and the base objects
contain initial values. A {\em step} is taken by some process whenever
the process executes a primitive on a shared register; the step may also include
some local computation that is performed before the execution of the primitive. 
An {\em execution} is a sequence of steps. The {\em interval contention} of an instance of some operation in an 
execution is the number of processes that are active during the execution of this instance. 
The {\em step complexity} of an operation is the maximum number of
steps that any thread performs during the execution of any instance of the operation in any execution.
{\em Wait-freedom}
guarantees that every process finishes each operation it executes in a finite number 
of steps.

\section{The L-UC Algorithm}
\label{sec:description_lsimopt}
This section presents \LSIMOPT, our wait-free universal construction for large objects.

\subsection{Overview}
\label{overview}

\begin{lstlisting} [float=!ht,xleftmargin=.04\textwidth,numbers=left,basicstyle=\scriptsize,language=c,breaklines=true,numberblanklines=false,caption={Data structures used in \LSIMOPT\ and pseudocode for \LSIMAP.},label={alg:lsimopt_ds},escapechar=@,name=lsimopt]
struct NewVar {    // node of list of newly allocated variables
    ItemSV *var;   // points to the ItemSV struct of the variable
    NewVar *next;  // points to the next element of the list
};

struct NewList {   
    ItemSV *first;
};

struct State {    
   boolean applied[1..n];
   boolean papplied[1..n];
   int seq;
   NewList *var_list;  
   RetVal RVals[1..n]; // return values
};

struct DirectoryNode {
   Name name;      // variable name
   ItemSV *sv;     // data item for the variable
   Value val;      // value of the data item
};

struct ItemSV {  // data item for a variable
   Value val[0..1];// old and new values of data item
   int toggle;     // toggle shows the current value of data item
   int seq;
};

// Toggles is implemented as an integer of @$n$@ bits; if @$n$@ is big, more than one such integers can be used
shared Integer Toggles = @$<0, ..., 0>$@; 
shared State S = @$<F,...,F>, <F,...,F>, 0, <\bot>, <\bot, ...,\bot>>$@;
shared OpType Announce[1..n] = {@$\bot$@, ..., @$\bot$@};

// Private local variable for process @$p_i$@
Integer @$toggle_i$@ = @$2^i$@;

RetVal ApplyOp(request req){  // Pseudocode for process @$p_i$@
    Announce[i] = req;        // Announce request @$req$@ @\label{alg:lsimopt:announce_op}@
    @$toggle_i$@ = -@$toggle_i$@; @\label{alg:lsimopt:toggle_toggle}@
    @\FAD@(Toggles, @$toggle_i$@);       // toggle @$p_i$@'s bit by adding @$2^i$@ to Toggles @\label{alg:lsimopt:first_add}@
    @\ATTEMPT@();      // call @\ATTEMPT@ twice to ensure that req will be performed@\label{alg:lsimopt:first_attempt}@
    @\ATTEMPT@();                 @\label{alg:lsimopt:second_attempt}@
    return S.rvals[i];        // @$p_i$@ finds its return value into @$S.rvals[i]$@
}
\end{lstlisting}

\begin{lstlisting} [float=!ht,xleftmargin=.0\textwidth,numbers=left,basicstyle=\scriptsize,language=c,breaklines=true,numberblanklines=false,caption={Pseudocode for \LSIMOPT.},
label={alg:lsimopt},escapechar=@,name=lsimopt-1, postbreak=\/\/\space, breakautoindent=true, breakindent=150pt, breaklines]
void Attempt(Request req) {              // pseudocode for process @$p_i$@
  ProcessIndex q, j;
  State ls, tmp;
  Set lact;
  Directory @$D$@;
  NewVar *pvar = new NewVar(), *ltop;
  ItemSV sv, *psv = new ItemSV();
  
  psv@$\rightarrow \langle$@val, toggle, seq@$\rangle$@ = @$<<\bot, \bot>,0,0>$@;
  pvar@$\rightarrow \langle$@var, next@$\rangle$@ = <psv, null>; 
  for j=1 to 2 do { @\label{alg:lsimopt:attempt_loop}@
     D = @$\emptyset$@;                                // initialize direcory D@\label{alg:lsimopt:dir_init}@
     ls = @\LL@(S);                           // create a local copy of @$S$@ @\label{alg:lsimopt:ll_iteration}@
     lact = Toggles;                       // read active set@\label{alg:lsimopt:read_toggles}@
     ltop = ls.var_list@$\rightarrow$@first; // read pointer to the list of newly-allocated variables@\label{alg:lsimopt:init_list}@
     tmp.seq = ls.seq + 1;                               @\label{alg:lsimopt:tmp_inc}@
     tmp.papplied[1..n] = ls.applied[1..n];              @\label{alg:lsimopt:s_papplied}@
     tmp.applied[1..n] = lact[1..n];       // @$p_i$@ will later attempt to update S with tmp, so it sets the fields of tmp appropriately@\label{alg:lsimopt:s_applied}@
     tmp.rvals[1...n] = ls.rvals[1..n];@\label{alg:lsimopt:copy_rvals}@
     for q=1 to n do {                     @\label{alg:lsimopt:for_loop}@
        if (ls.applied[q] @$\neq$@ ls.papplied[q]) { // q's request is pending@\label{alg:lsimopt:if_apply}@
           foreach access of a variable x while applying request Announce[q]{@\label{alg:lsimopt:foreach_access}@
              if (x is a newly allocated variable) {@\label{alg:lsimopt:alloc_var}@
                 if(@\CAS@(ltop@$\rightarrow$@next, null, pvar)){@\label{alg:lsimopt:add_list}@
                    psv = new ItemSV();
                    psv@$\rightarrow \langle$@val, toggle, seq @$\rangle$@ = @$<<\bot, \bot>,0,0>$@;
                    pvar = new NewVar();
                    pvar@$\rightarrow \langle$@var, next@$\rangle$@ = <psv, null>;
                 } @\Suppressnumber@
                 // use node pointed by @$ltop \rightarrow next$@ as the new variable's metadata@\Reactivatenumber@
                 ltop = ltop@$\rightarrow$@next;@\label{alg:lsimopt:new_var}@
                 add <x, ltop@$\rightarrow$@var, ltop@$\rightarrow$@var.val[0]> to D; @\label{alg:lsimopt:add_new_item_to_dir}@
              } else {                  // x is not a newly allocated variable 
                 let svp be a pointer to the ItemSV struct for x; 
                 if (this access is a read instruction) {  @\label{alg:lsimopt:perfom_read}\Suppressnumber@
                    // perform the request on the local copy of x (if any) @\Reactivatenumber@
                    if (x exists in D) read x from D;
                    else {
                       sv = @\LL@(*svp);@\label{alg:lsimopt:ll_dir}@
                       if (tmp.seq==sv.seq) add <x,svp,sv.val[1-sv.toggle]> to D;@\label{alg:lsimopt:add_dir1}@
                       else if(tmp.seq>sv.seq) add <x,svp,sv.val[sv.toggle]> to D;@\label{alg:lsimopt:add_dir2}@
                       else goto Line @\ref{alg:lsimopt:vl}@; // values read from @$S$@ by @$p_i$@ obsolete, so start from scratch@\label{alg:lsimopt:obsolute_on_read}@
                    }
                 } else if (the access is a write instruction) update x in D;@\label{alg:lsimopt:update_dir}@
              }
           }
           store into tmp.rvals[q] the return value;@\label{alg:lsimopt:calculate_return_value}@
        }
     }
     if (!@\VL@(S)) continue; // value read in @$S$@ by @$p_i$@ is obsolete, so start from scratch @\label{alg:lsimopt:vl}@
     foreach record <x, svp, v> in D {@\label{alg:lsimopt:flush_dir}@
        if(svp@$\rightarrow$@seq > tmp.seq) break;  // all requests have been applied, so leave the loop @\label{alg:lsimopt:for_break}@
        else if(svp@$\rightarrow$@seq == tmp.seq) continue; // the variable has been modified, so continue @\label{alg:lsimopt:sc_dir1}@
        else if(svp@$\rightarrow$@toggle == 0) SC(*svp, @$<<$@svp@$\rightarrow$@val[0],v>, 1, tmp.seq>); @\label{alg:lsimopt:sc_dir2}@
        else SC(*svp, @$<<$@v, svp@$\rightarrow$@val[1]>, 0, tmp.seq>); // make update visible@\label{alg:lsimopt:sc_dir3}@
     }
     tmp.var_list = new List();  tmp.var_list@$\rightarrow$@first = null; @\label{alg:lsimopt:new_list}@
     @\SC@(S, tmp);                       // try to modify S @\label{alg:lsimopt:sc_on_s}@
  }
}
\end{lstlisting}

The pseudocode for \LSIMOPT\ is provided in Listings~\ref{alg:lsimopt_ds} and~\ref{alg:lsimopt}.
The state of the simulated data structure in \LSIMOPT\ is shared and it can be updated directly by any process. 
Each process $p$ that wants to apply a request,
first announces it in an $Announce$ array. 
In addition to the $Announce$ array, 
\LSIMOPT\ uses a bit vector $Toggles$ of $n$ bits, one for each process. 
A process $p_i$ toggles its bit, $Toggles[i]$, after announcing a new request.
The use of $Toggles$ implements a fast mechanism for informing other processes of those processes that 
have pending requests.

Each execution of \LSIMOPT\ can be partitioned into phases.
In each phase $i \geq 1$, the set of requests that will be executed in the next phase
is agreed upon by the processes that are active.
Moreover, those requests that have been agreed upon in the previous phase are indeed executed.

A process $p_i$ that wants to execute a new request, it first announces it in $Announce$,
and then it toggles its bit in $Toggles$. Afterwards, it calls a function, called \ATTEMPT, twice: 
After the execution of the first instance of \ATTEMPT\ by $p_i$, it is ensured
that the set of requests agreed upon in one of the phases that overlap
the execution of the \ATTEMPT, contains $p_i$'s request. After the execution of the
second instance of \ATTEMPT\ by $p_i$, it is ensured that 
$p_i$'s request has been applied. 

\LSIMOPT\  uses an \LC\ object $S$ which stores appropriate fields
to ensure the required synchronization between the processes in each phase. 
The first phase  (phase 1) starts at the initial configuration and ends when the
first successful \SC\ is applied on $S$. Phase $i > 1$ starts when 
phase $i-1$ finishes and ends when the $i$-th successful \SC\ is applied on $S$.

To decide which set of requests will be executed in each phase, 
$S$ contains two bit vectors, called $applied$ and $papplied$, of $n$ bits each
(one for each process). The current request initiated by a process $p_i$ 
has not yet been applied, if $S.applied[i] \neq S.papplied[i]$.
When this condition holds, we call the current request of process $p_i$ {\em pending}.

In each instance of \ATTEMPT, $p_i$ 
copies the value of $S$ in a local variable $ls$ (line~\ref{alg:lsimopt:ll_iteration}), 
records necessary changes that it makes to its fields
in another local variable $tmp$ 
(lines~\ref{alg:lsimopt:tmp_inc}-\ref{alg:lsimopt:copy_rvals}, \ref{alg:lsimopt:calculate_return_value}, \ref{alg:lsimopt:new_list}), 
and uses \SC\ in an effort to update $S$ to the value contained in $tmp$ (line~\ref{alg:lsimopt:sc_on_s}).
Specifically, $p_i$ reads $S$ on line~\ref{alg:lsimopt:ll_iteration} (by performing an \LL)
and $Toggles$ on line~\ref{alg:lsimopt:read_toggles}. It then copies $S.applied$ into 
$tmp.papplied$ (line~\ref{alg:lsimopt:s_papplied}) and $Toggles$ into $tmp.applied$ (line~\ref{alg:lsimopt:s_applied}).
Recall that the $applied$ and $papplied$ fields of $S$ 
encode the requests that are to be performed in each phase.
So, if the \SC\ that $p_i$ performs on line~\ref{alg:lsimopt:sc_on_s} succeeds, 
all processes that will read the value this \SC\ will write to $S$, 
will attempt to perform the requests encoded by $p_i$ in those fields.

Next, for each $j$, $1 \leq j \neq n$, $p_i$ checks 
whether $ls.applied[j] \neq ls.papplied[j]$  (lines~\ref{alg:lsimopt:for_loop}-\ref{alg:lsimopt:if_apply}), 
and if this is so, it applies the request recorded in $Announce[j]$.
To execute the pending requests recorded in $S$,
a process $p_i$ uses a caching mechanism as 
in~\cite{B93, CER10}: 
When a process first accesses a shared variable 
(e.g., a variable of the simulated shared data structure),
it maintains a copy of it in a directory, $D$ (which is local to $p_i$).
For each pending request recorded in $S$, 
the required updates are first performed by 
$p_i$ in the local copies of the data items that are residing in the directory
(lines \ref{alg:lsimopt:foreach_access}-\ref{alg:lsimopt:calculate_return_value}).
Read requests executed by $p_i$ are also served using $D$.
Only after 
it has finished the simulation of all pending requests, 
$p_i$ applies the changes listed in the elements of its directory to the shared data structure 
(lines~\ref{alg:lsimopt:flush_dir}-\ref{alg:lsimopt:sc_dir3}).

For each data item $x$ of the simulated object's state, \LSIMOPT\ maintains a 
record (struct) of type $ItemSV$. 
This struct stores the old and the current value of the data item in an array $val$
of two elements, a toggle bit that identifies the position in the $val$ array 
from where the current value for $x$ should be read, and a sequence number 
that is used for synchronization.

Note that $S$ contains also a field $seq$ that is incremented every time a successful \SC\ on $S$ is performed.
This field identifies the current phase of the execution.
Before performing an update on the shared data structure (lines~\ref{alg:lsimopt:flush_dir}-\ref{alg:lsimopt:sc_dir3}), 
$p_i$ validates the values of the $seq$ field read in $S$ ($tmp.seq$) and that stored in $ItemSV$ for $x$ ($svp \rightarrow seq$). 
Only if $svp \rightarrow seq < tmp.seq$ (line~\ref{alg:lsimopt:sc_dir3}), the update is performed since otherwise
it is already obsolete, i.e., $S.seq$ is already greater than 
$tmp.seq$ and therefore the \SC\ of line~\ref{alg:lsimopt:sc_on_s} by $p_i$ will fail.

Both the old and the current values of $x$ must be stored  in $ItemSV$
in order to avoid the following bad scenario. Consider two processes $p_i$ and $p_j$ that simulate the 
same request $req$. Assume that $p_i$ is  ready to execute
line~\ref{alg:lsimopt:ll_dir} for some variable $x$, 
whereas $p_j$ has finished the simulation of $req$ (lines \ref{alg:lsimopt:flush_dir}-\ref{alg:lsimopt:sc_dir3}) 
and has started updating the shared data structure. Then, it might happen that $p_i$ 
reads the updated version for $x$ although it should have read the old version. 
For this reason, $p_j$ stores the old value (in addition to the new value)
in one of the entries of the $val$ array and appropriately updates the toggle bit 
to indicate which of the two values is the new one. 
If $p_i$ discovers that it is too slow
(line~\ref{alg:lsimopt:add_dir1}), it reads the old value for $x$ stored in the $1-toggle$ 
entry of its $val$ array. Notice that, to ensure wait-freedom, 
$p_i$ should continue executing $req$ (to cope with the case that $p_j$ 
fails before performing all the required updates to the shared data structure).

When a new data item $x$ is allocated while executing a set of requests, additional synchronization 
between the processes that execute this set of requests is required to avoid situations 
where several processes allocate, each, a different record for $x$.
We use a technique similar to that presented in~\cite{CER10} to 
ensure that all these processes use the same allocated ItemSV structure for $x$. 
Specifically, \LSIMOPT\ stores into $S$ a pointer (called $var\_list$) to a list of newly 
created data items shared by all processes that read this instance of $S$. Each time a process $p_i$ 
needs to allocate the $j$-th, $j \geq 1$, such data item, it tries to add a structure of type $NewVar$ 
as the $j$-th element of the list (line~\ref{alg:lsimopt:add_list}). If it does not succeed, some other
process has already done so, so $p$ uses this structure (by moving pointer $ltop$ to this element on
line~\ref{alg:lsimopt:init_list}, and by inserting $ltop \rightarrow var$ in its dictionary on
line~\ref{alg:lsimopt:add_new_item_to_dir}).

We remark that 
the fields of $ItemSV$ must be updated in an atomic way using \SC. This requires that registers
in the system store two words which is impractical. However, we can utilize single-word 
registers by using indirection. Indirection can also be used to implement $S$ using single-word registers.

\subsection{Detailed Description of \ATTEMPT}
\label{detailed description}
In the following, we detail the implementation of function \ATTEMPT, presented in Algorithm~\ref{alg:lsimopt}.
When \ATTEMPT\ is executed by some process $p_i$, $p_i$ executes two iterations (line~\ref{alg:lsimopt:attempt_loop}) of checking whether 
there are pending requests and of attempting to apply them, as follows. It initializes its local directory 
$D$ (line~\ref{alg:lsimopt:dir_init}), creates in $ls$ a local copy of the state of the simulated object 
(line~\ref{alg:lsimopt:ll_iteration}), and reads in $lact$ the value of $Toggles$ (line~\ref{alg:lsimopt:read_toggles}), 
thus obtaining a view of which processes have pending requests  at the current point in time
(i.e., calculating the set of pending requests).
Furthermore, it 
locally stores into $ltop$ a pointer to the current variable list of the simulated object (line~\ref{alg:lsimopt:init_list}). 
Recall that the state of the object is copied into local variable $ls$ using an \LL\ primitive. 
In case this instance of \ATTEMPT\ is successful in applying the pending requests, it will update the shared state of 
the system using an \SC\ primitive. For this purpose, the local variable $tmp$ is prepared in lines~\ref{alg:lsimopt:tmp_inc}~to~\ref{alg:lsimopt:s_applied}, 
to serve as the value that will be stored into the shared state in case of success. 

After having read the state of the simulated object, as well as the state of the requests of the other processes, 
$p_i$ can detect which requests are pending. For this purpose, it iterates through the (locally stored) state 
of each process (line~\ref{alg:lsimopt:for_loop}) and checks whether the  values of $papplied$ and $applied$ differ for this process 
(line~\ref{alg:lsimopt:if_apply}). If so, the request of this process was still pending when \ATTEMPT\ read the 
value of $Toggles$ and therefore, \ATTEMPT\ intents to apply it. Notice that the iteration through the $papplied$ 
and $applied$ values consist of local steps. Notice also that at most $k$ out of $n$ processes have active 
requests, meaning what the request application contributes to step complexity depends on $k$ rather than $n$. 

We remark that the request of a process is expressed as a piece of sequential code. Therefore, 
in order to apply the request of some process, an instance of \ATTEMPT\ has to run through the sequential code 
of this request and carry out the variable accesses that this request entails, i.e. \ATTEMPT\ has to apply the 
modifications that this request incurs on the simulated object's variables (line~\ref{alg:lsimopt:foreach_access}). 
We distinguish three cases, namely the case where an access creates a new variable, the case where an access 
reads a variable, and the case where an access modifies an already existing variable. 

In the first case (line~\ref{alg:lsimopt:add_list}), the new variable, which was created and stored in local 
variable $pvar$, must be added to the shared list of variables of the simulated object. Recall that a pointer 
to the top of the variable list has been read by $p_i$ and stored in local variable $ltop$.
Recall also that all processes 
are trying to perform the announced requests in the same order. As with each instance of \ATTEMPT, so also the $p_i$ instance of
 \ATTEMPT\ tries to add $pvar$ 
to the top of the list using a \CAS\ primitive (line~\ref{alg:lsimopt:add_list}). In case this is successful, 
the metadata of this variable is initialized. In case the \CAS\ is unsuccessful, then some other process has 
updated the object's variable list since this instance of \ATTEMPT\ read it into $ltop$. 
Given that all processes follow the same order when trying to insert newly-allocated variables, the failure means that the variable has already been inserted in the shared list of variables of the simulated object. In either case, i.e. either successful or unsuccessful insertion by $p_i$,
$ltop$ is updated to point to the data item of the newly allocated variable.
Furthermore, the newly added variable is included into 
the local variable dictionary (line~\ref{alg:lsimopt:add_new_item_to_dir}).

In the second case (line~\ref{alg:lsimopt:perfom_read}), the access to be performed is a read to a variable 
of the simulated object. If \ATTEMPT\ already has a local copy of this variable in its dictionary, it reads 
the value from there. If no local copy is present in the dictionary (line~\ref{alg:lsimopt:ll_dir}), then 
\ATTEMPT\ reads the variable using an \LL\ primitive (line~\ref{alg:lsimopt:ll_dir}). At the same time, it 
checks the sequence number of the value that it has read, and in case this sequence number is less or equal 
to the local sequence number stored in $tmp$, then \ATTEMPT\ considers that it is reading a valid value. 
This value is then added to the local dictionary. However, in case the variable's sequence number is larger 
than the local sequence number, this hints that this instance of \ATTEMPT\ has been rendered obsolete by some other 
process that has already applied all requests that this instance of \ATTEMPT\ is applying. In order to find out if this 
is the case, \ATTEMPT\ verifies whether the state of $S$ has changed since it last read it (line~\ref{alg:lsimopt:vl})
and if so, it gives up the current iteration of the for loop of line~\ref{alg:lsimopt:attempt_loop}.

Finally, in the third case (line~\ref{alg:lsimopt:update_dir}), where the access is a write to an already 
existing variable.
In case that the accessed variable already exists in the local dictionary, 
the update on the local dictionary (line~\ref{alg:lsimopt:update_dir}), 
updates the variable's value stored in the local dictionary.
Otherwise, the update (line~\ref{alg:lsimopt:update_dir})
creates a new entry and stores the value of the variable.
Once the sequential code for the current request has all been run through and all variable accesses for 
the request have been performed, the request returns a return value, which is stored by \ATTEMPT\ for the 
process to access (line~\ref{alg:lsimopt:calculate_return_value}).

Recall that any update to a variable of the simulated object is performed locally by \ATTEMPT. 
Therefore, once all active requests have been applied, \ATTEMPT\ has to write back the local 
updates to the shared variables of the simulated object (lines~\ref{alg:lsimopt:flush_dir}~-~\ref{alg:lsimopt:sc_dir3}). 
Notice that once again, the sequence numbers of the local and shared copies are instrumental 
in detecting whether a variable has already been updated or not 
(lines~\ref{alg:lsimopt:sc_dir1}~-~\ref{alg:lsimopt:sc_dir3}).
More specifically, the condition of line~\ref{alg:lsimopt:sc_dir1} checks if another process has already
updated or not the value of the shared variable while trying to apply the same set of operations 
calculated in lines~\ref{alg:lsimopt:s_papplied}~-~\ref{alg:lsimopt:copy_rvals}.
In case that a process is very slow and the whole set of operations calculated in
lines~\ref{alg:lsimopt:s_papplied}~-~\ref{alg:lsimopt:copy_rvals} is applied, the condition of 
line~\ref{alg:lsimopt:sc_dir2} fails, and the process breaks the execution (line~\ref{alg:lsimopt:for_break})
of the \textit{for-loop} of lines~\ref{alg:lsimopt:flush_dir}~-~\ref{alg:lsimopt:sc_dir3}.
Finally, once the updates have been performed, \ATTEMPT\ tries to update $S$, before performing any remaining 
iteration of the for loop of line~\ref{alg:lsimopt:attempt_loop}.  

\subsection{Step Complexity}
\label{step complexity}
By inspection of the pseudocode of \APPLYOP, it becomes apparent that its step complexity
is determined by the step complexity of \ATTEMPT. In a practical version of \LSIMOPT\
where $S$ is implemented using indirection, lines~\ref{alg:lsimopt:ll_iteration} and~\ref{alg:lsimopt:read_toggles} 
contribute $O(n)$ to performance, since the size of the data records that are read is $O(n)$.
The body of the if statement of line~\ref{alg:lsimopt:if_apply} (i.e., lines~\ref{alg:lsimopt:foreach_access}-\ref{alg:lsimopt:update_dir}) 
is executed $O(k)$ times, each  time contributing a factor of $O(w)$ (because of the foreach statement of line~\ref{alg:lsimopt:foreach_access}). 
Note that searching an element in the dictionary, adding an element to it
or removing an element from it does not cause any shared memory accesses, 
i.e., it causes only local computation. So, the cost of executing lines~\ref{alg:lsimopt:alloc_var}-\ref{alg:lsimopt:calculate_return_value} is $O(1)$.
Note also that at most $O(kw)$ elements are contained in each dictionary. 
Therefore, the \textit{foreach} of line~\ref{alg:lsimopt:flush_dir} contributes $O(kw)$ to the total cost. The rest
of the code lines access only local variables and thus they do not contribute to the step complexity
of the algorithm.  We conclude that the step complexity of \APPLYOP\ is $O(n+kw)$.

\subsection{Correctness Proof}
\label{proof:lsimopt}

This section provides a sketch of the correctness proof of \LSIMOPT. We start with some useful notation.
Let $\alpha$ be any execution of \LSIMOPT\ and assume that some thread $p_i$, $i \in \{1, ..., n\}$,
executes $m_i > 0$ requests in $\alpha$.
Let $req_{j}^i$ be the argument of the $j$-th call of \LSIMOPT\ by $p_i$
and let $\pi_{j}^i$ be the $j$-th instance of \ATTEMPT\ executed by $p_i$ (Figure~\ref{fig:lsimopt_execution}).
Let $C_0$ be the initial configuration. 
Define as $Q_j^i$ the configuration after the execution of the \FAD\ instruction of line~\ref{alg:lsimopt:first_add}; let $Q_0^i = C_0$. 
We use $Toggles[i]$, $i \in \{1, \ldots, n\}$, to denote the $i$-th bit of $Toggles$,
and let $toggle_j^i$ be the value of $p_i$'s local variable $toggle_i$ at the end of $req_j^i$.

In the following lemma, we argue that during the execution of each of the two iterations of the for loop of 
line~\ref{alg:lsimopt:attempt_loop} of any instance of \ATTEMPT,  at least one successful \SC\ instruction 
is performed. 

\begin{lemma}
\label{lsimopt:lemma:two_sc}
Consider any $j$, $0 < j \leq m_i$.
There are at least two successful \SC\ instructions
in the execution interval of $\pi_j^i$.
\end{lemma}

We continue with two technical lemmas. The first argues that
the value of $p_i$'s bit in the $Toggles$ array is equal to $j \mod 2$
after the execution of the $j$-th \FAD\ instruction of line~\ref{alg:lsimopt:first_add} 
by $p_i$.
It also shows that
no process other than $p_i$ can change this bit.
%Due to lack of space, the proofs of these lemmas (as well as of others that follow) 
%are presented in the attached appendix.

\begin{lemma}
\label{lsimopt:col:toggles_intervals}
For each $j$, $0 \leq j \leq m_i$, it holds that
(1) $Toggles[i] = j \mod 2$ at $Q_j^i$, and
(2) $Toggles[i]$ has the same value between $Q_{j-1}^i$ and $Q_{j}^{i}$.
\end{lemma}

The next lemma studies the value of $S.applied[i]$ after the execution 
of the $j$-th instance of \ATTEMPT\  by $p_i$.

\begin{lemma}
\label{lsimopt:lemma:end_of_attempt}
Consider any execution $\pi_j^i$, $j > 0$, of function \ATTEMPT\ by some thread $p_i$.
$S.applied[i]$ is equal to $v = \lceil j/2 \rceil \mod 2$ just after the end of $\pi_{j}^i$.
\end{lemma}

For each $l > 0$, let $C_l$ be the configuration resulting after the execution 
of the $l$-th \FAD\ instruction in $\alpha$.
At $C_0$, $S.applied[i]$ is
equal to \FALSE. Lemma~\ref{lsimopt:lemma:end_of_attempt}
implies that just after $\pi_1^i$, $S.applied[i]$ is equal to \TRUE.
Let $C_1^i$ be the first configuration between $C_0$ and the end of
$\pi_1^i$ at which $S.applied[i]$ is equal to \TRUE.
%Lemma \ref{lsimopt:lemma:end_of_attempt} implies that just after 
%the end of $\pi_3^i$, $S.applied[i]$ is equal to \FALSE. 
%Let $C_2^i$ be the first configuration after $C_1^i$ such 
%that $S.applied[i]$ is equal to \FALSE. 
%Obviously, $C_2^i$ precedes the end of $\pi_3^i$.
Consider any request $req_j^i$, $j > 1$. 
Lemma \ref{lsimopt:lemma:end_of_attempt} implies
that just after $\pi_{2j-2}^i$, $S.applied[i]$ is equal to 
$\lceil (j-2)/2 \rceil \mod 2 = (j - 1) \mod 2$,
while just after $\pi_{2j-1}^i$, $S.applied[i]$ is equal to
 $\lceil (2j - 1)/2 \rceil \mod 2 = j \mod 2 \neq (j - 1)\mod 2$.
Let $C_{j}^i$ be the first configuration between the end of
$\pi_{2j-2}^i$ and the end of $\pi_{2j-1}^i$ such that $S.applied[i]$
is equal to $j\mod 2$. 
%Similarly, let $C_{j}^i$ be the first configuration after 
%the end of $\pi_{2j-2}^i$ such that $S.applied[i]$ is equal to $j\mod 2$. 
Figure~\ref{fig:lsimopt_execution} illustrates the above notation.

Since the value of $S.applied[i]$ can change only by the
execution of \SC\ instructions on $S$, it follows that just before
$C_{j-1}^i$ a successful \SC\ on $S$ is
%$C_j^i$ a successful \SC\ on $S$ is
executed. Let $SC_{j}^i$ be this \SC\ instruction 
and let $LL_{j}^i$ be its matching \LL\ instruction.
Let $T_j^i$ be the read of $Toggles$ that is executed between
$LL_j^i$ and $SC_j^i$ by the same thread.

\begin{figure}[!t]
\begin{center}
\includegraphics[width=0.99\textwidth]{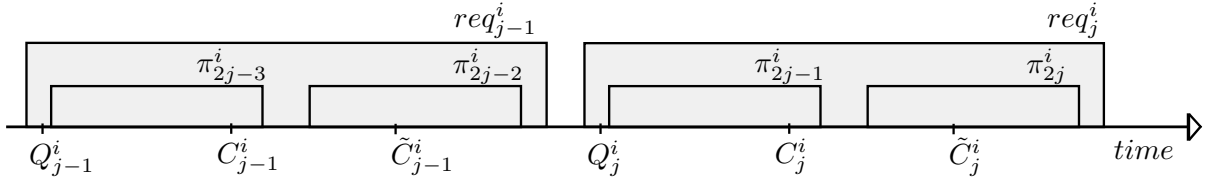}
\end{center}
\vspace{-0.2in}
\caption{An example of an execution of \LSIMOPT.}
\label{fig:lsimopt_execution}
\end{figure}

Lemma~\ref{lsimopt:lemma:T_follows_Q} states that $T_j^i$ 
%the read of $Toggles$ during the execution of the $j$-th request of process $p_i$
is performed at the proper timing and returns the anticipated value. 

\begin{lemma}
\label{lsimopt:lemma:T_follows_Q}
Consider any $j$, $0 < j \leq m_i$, it holds that $T_j^i$ is executed 
after $Q_{j}^i$ and reads $j\mod 2$ in $Toggles[i]$.
\end{lemma}

\begin{proof}
Assume, by the way of contradiction, that $T_j^i$ is executed before $Q_j^i$.
Let $\pi_x$ be the \ATTEMPT\ that executes $T_j^i$.

Assume first that $j = 1$.
Then, by its definition, $SC_1^i$ (which is executed by $\pi_x$ after $T_1^i$)
writes to $S \rightarrow applied[i]$ a value equal to $\lceil j/2 \rceil \mod 2$;
the code (lines \ref{alg:lsimopt:read_toggles}, \ref{alg:lsimopt:s_applied}) implies that, 
in this case, $T_1^i$ reads $1$ in $Toggles[i]$.
Lemma \ref{lsimopt:col:toggles_intervals} implies that $Toggles[i] = 0$ between $C_0$
and $Q_1^i$. Thus, $T_1^i$ could not read $1$ in $Toggles[i]$, which is a contradiction.

Assume now that $j > 1$.
By our assumption that $T_j^i$ is executed before $Q_j^i$, it follows that $LL_j^i$, which is executed
before $T_j^i$, precedes $Q_j^i$. In case that $T_j^i$ follows $Q_{j-1}^i$, Lemma
\ref{lsimopt:col:toggles_intervals} implies that $T_j^i$ reads 
$(j-1) \mod 2 \neq j \mod 2$ in $Toggles[i]$.
By the pseudocode (lines \ref{alg:lsimopt:read_toggles}, \ref{alg:lsimopt:s_applied} and
\ref{alg:lsimopt:sc_on_s}), it follows that $\pi_x$ writes the value $(j-1)\mod 2$
into $S.applied[i]$. By its definition, $SC_j^i$ stores $j\mod 2$
into $S.applied[i]$, which is a contradiction.
Thus, $T_j^i$ is executed before $Q_{j-1}^i$. By its definition, $\pi_{2j-3}^i$ starts
its execution after $Q_{j-1}^i$ and finishes its execution before $C_j^i$.
Lemma~\ref{lsimopt:lemma:two_sc} implies that at least two successful
\SC\ instructions are executed in the execution interval of $\pi_{2j-3}^i$.
Recall that $LL_j^i$ precedes $T_j^i$ and therefore also the beginning of $\pi_{2j-3}^i$,
while by definition $SC_j^i$ follows the end of $\pi_{2j-3}^i$. It follows that $SC_j^i$
is not a successful \SC\ instruction, which is a contraction.
\end{proof}

We next argue that, between certain configurations (namely  $C_{j-1}^i$ and $C_{j}^i$), 
the value of $S.applied[i]$ has the anticipated value and this value does not change
in the execution interval defined by the two configurations. 

\begin{lemma}
\label{lsimopt:lemma:applied_remains_unchanged}
Consider any $j$, $0 < j \leq m_i$. At each configuration
$C$ between $C_{j-1}^i$ and $C_{j}^i$,
it holds that  $S.applied[i] = (j-1)\mod 2$.
\end{lemma}

\begin{proof}
Assume, by the way of contradiction, that there is at least one configuration
between $C_{j-1}^i$ and $C_j^i$ such that $S \rightarrow applied[i]$ 
is equal to some value $v_x \neq (j-1)\mod 2$. Let $C_x$ be the first of these configurations.
Since only \SC\ instructions of line~\ref{alg:lsimopt:sc_on_s} write on base object $S$, 
it follows that there is a successful \SC\ instruction, let it be $SC_x$, executed just 
before $C_x$ that stores $v_x$ at $S.applied[i]$.
Let $\pi_x$ be the \ATTEMPT\ that executes $SC_x$ and let $T_x$ be the read
instruction that $\pi_x$ executes on line~\ref{alg:lsimopt:read_toggles} of the pseudocode.
By the definition of $C_{j-1}^i$ and $Q_{j-1}^i$, it is implied that $C_{j-1}^i$ follows 
$Q_{j-1}^i$ and precedes $Q_j^i$. 
Lemma~\ref{lsimopt:col:toggles_intervals} implies that $Toggles[i] = (j-1)\mod 2 \neq v_x$
in any configuration between $Q_{j-1}^i$ and $Q_j^i$.
Since $SC_x$ writes $v_x$ into $S.applied[i]$, the pseudocode (lines~\ref{alg:lsimopt:read_toggles} 
and~\ref{alg:lsimopt:sc_on_s}) imply that $T_x$ precedes $Q_{j-1}^i$.
It follows that $LL_x$ precedes $Q_{j-1}$, since $LL_x$ precedes $T_x$. Therefore $LL_x$ precedes $C_{j-1}$. 
This implies that there is a successful \SC\ instruction, which is $SC_{j-1}^i$,
between $LL_x$ and $SC_x$. Thus, $SC_x$ is a failed \SC\ instruction, which is a contradiction.
\end{proof}

By Lemma~\ref{lsimopt:lemma:applied_remains_unchanged} and the pseudocode 
(line \ref{alg:lsimopt:s_papplied}), it follows that $S.papplied[i] = 1 - (j\mod 2)$ at $C_{j}^i$.
Denote by $\tilde C_j^i$ the first configuration after $C_j^i$ such that a successful 
\SC\ instruction is executed. 

The next lemma studies properties of $\tilde C_j^i$.

\begin{lemma}
\label{lsimopt:lemma:position_of_tilde_C_j^i}
$\tilde C_{j-1}^i$ precedes $C_j^i$ and follows $C_{j-1}^i$. 
\end{lemma}

\remove{
\begin{proof}
By the definition of $\tilde C_{j-1}^i$, it is implied that $\tilde C_{j-1}^i$ follows $C_{j-1}^i$.
Lemma~\ref{lsimopt:lemma:T_follows_Q} implies that $C_j^i$ follows $Q_j^i$.
By its definition, $Q_j^i$ follows $\pi_{2j-2}$.
By Lemma~\ref{lsimopt:lemma:end_of_attempt}, it follows that $C_{j-1}^i$ precedes the end of 
$\pi_{2j-3}$. Thus, $\pi_{2j-2}^i$ begins its execution after $C_{j-1}^i$ and ends its
execution before $C_j^i$. 
By Lemma~\ref{lsimopt:lemma:two_sc}, there are at least two successful
\SC\ instructions in the execution interval of $\pi_{2j-2}^i$. 
Since, $\tilde C_{j-1}^i$ is the first successful \SC\ just after $C_{j-1}^i$,
it follows that $\tilde C_{j-1}^i$ precedes the end of $\pi_{2j-2}^i$.
Therefore, $\tilde C_{j-1}^i$ precedes $C_j^i$.
\end{proof}
}

We next argue that the $applied$ and $papplied$ arrays of $S$ indicate
that $p_i$ does not have a pending request between $\tilde C_{j-1}^i$ and 
$C_j^i$.

\begin{lemma}
\label{lsimopt:lemma:relation_applied_papplied}
$S.papplied[i] = S.applied[i]$ in any configuration between $\tilde C_{j-1}^i$ and 
$C_j^i$ ($C_j^i$ is not included).
\end{lemma}

\remove{
\begin{proof}
By Lemma~\ref{lsimopt:lemma:applied_remains_unchanged}, it follows that $S.applied[i] = (j - 1) \mod 2$
between $C_{j-1}^i$ and $C_j^i$.
Assume by the way of contradiction that there at least one configuration between $\tilde C_{j-1}^i$
and $C_j^i$ such that $S.papplied[i] \neq (j - 1) \mod 2$ and let $C_x$ be the first of them.
Let $SC_x$ be the \SC\ instruction executed just before $C_x$ and let $LL_x$ be its matching 
\LL\ instruction. Since, $SC_x$ is a successful \SC\ instruction, $LL_x$ follows $\tilde C_{j-1}^i$.
By Lemma~\ref{lsimopt:lemma:applied_remains_unchanged}, it follows that $S.applied[i] = (j-1) \mod 2$ 
between  $C_{j-1}^i$ and $C_j^i$. Thus, $LL_x$ reads $(j-1)\mod 2)$ in $S.applied[i]$. 
The pseudocode (lines~\ref{alg:lsimopt:ll_iteration} and~\ref{alg:lsimopt:s_papplied}) implies 
that the \SC\ instruction at $\tilde{C}_j^i$ stores a value equal to $(j-1)\mod 2$ into $S.papplied[i]$, 
which is a contradiction.
\end{proof}
}

By Lemma~\ref{lsimopt:lemma:relation_applied_papplied}, and 
by line~\ref{alg:lsimopt:s_papplied}, it follows that $S.papplied[i] = 1 - S.applied[i]$ at $C_j^i$. 
This and the definition of $\tilde C_j^i$ imply:

\begin{lemma}
\label{lsimopt:obs:relation_applied_papplied}
$S.papplied[i] = 1 - S.applied[i]$ in any configuration between $C_j^i$ and $\tilde C_j^i$
	($\tilde C_j^i$ is not included).
\end{lemma}

We continue to define what it means for a process to apply a request on the simulated object.
We say that a request $req$ by some thread $p_i$ {\em is applied} on the simulated object if
(1) the \READ\ instruction on $Toggles$ (line~\ref{alg:lsimopt:read_toggles}),
executed by some request  $req'$ (that might be $req$ or any other request), 
includes $p_i$ in the set of threads it returns,
(2) procedure \ATTEMPT, executed by $req'$ reads in $Announce[i]$, the request type written there by 
$p_i$ for $req$ and considers it as the new request type for $p_i$,
(3) \ATTEMPT\ by $req'$ calls {\tt apply} for $req$ 
(lines \ref{alg:lsimopt:foreach_access} - \ref{alg:lsimopt:calculate_return_value}),
and the execution of the \SC\ at line \ref{alg:lsimopt:sc_on_s} (let it be $SC_r$)
on $S$ succeeds.
When these conditions are satisfied, we sometimes also say that
$req'$ applies $req$ on the simulated object or that $SC_r$ applies $req$ on the simulated object.

\begin{lemma}
\label{lsimopt:lemma:req_applied_at_least_once}
$req_j^i$ is applied to the simulated object at configuration $C_{3j-1}^i$.
\end{lemma}

\begin{proof}
Let $p_h$ be the \ATTEMPT\ that executes the successful \SC\ instruction (let it be $SC_h$
this \SC\ instruction) just before $\tilde C_j^i$. Let $LL_h$ be the matching \LL\
of $SC_h$. Since, $SC_h$ is a successful \SC\ instruction, it is implied that $LL_h$
follows $C_j^i$. Observation~\ref{lsimopt:obs:relation_applied_papplied} implies that
$LL_h$ reads for $S.applied[i]$ a value different from that stored in $S.papplied[i]$.
Therefore, the if statement of line~\ref{alg:lsimopt:if_apply} returns \TRUE.
Thus, a request for thread $p_i$ is applied at $\tilde C_j^i$.
Let $req'$ be this request and assume, by the way of contradiction,
that $req' \neq req_j^i$. Lemma \ref{lsimopt:lemma:T_follows_Q} implies that $\pi_h$
executes its read $T_h$ on $Toggles$ after $Q_j^i$. By the pseudocode 
(lines~\ref{alg:lsimopt:read_toggles}, \ref{alg:lsimopt:foreach_access}), 
$\pi_h$ reads $Announce[i]$ after $T_h$, thus the reading of $Announce[i]$ by $\pi_h$ 
is executed between $Q_j^i$ and $\tilde C_j^i$. Since $req_j^i$ writes its
request to $Announce[i]$ before $Q_j^i$, the reading of $Announce[i]$ by
$\pi_h$ returns $req_j^i$. Thus, $\pi_h$ applies $req_j^i$ as the request
of $p_i$ in the simulated object.
\end{proof}

We are now ready to assign linearization points.
For each $i \in \{1, ..., n\}$ and $j \geq 1$, we place the linearization point of
$req_j^i$ at $\tilde C_j^i$; ties are broken by the order
imposed by identifiers of threads.

It is not difficult to argue that the linearization point of each request
is placed in the execution interval of the request.

\begin{lemma}
Each request $req_j^i$ is linearized within its execution interval.
\end{lemma}

\remove{
\begin{proof}
Lemma~\ref{lsimopt:lemma:T_follows_Q} implies that $Q_j^i$
precedes $C_j^i$. Thus $Q_j^i$ precedes $\tilde C_j^i$.
Since $\tilde C_j^i$ is the first successful \SC\ on $S$ after $C_j^i$
and (by its definition and by Lemma~\ref{lsimopt:lemma:end_of_attempt}) 
$C_j^i$ precedes the end of $\pi_{2j-1}^i$,
$\tilde C_j^i$ precedes the end of $\pi_{2j}^i$. Thus, $\tilde C_j^i$ 
is in the execution interval of $req_j^i$. 
Thus, $req_j^i$ is executed in its execution interval.
\end{proof}
}

To prove consistency, denote by $SC_l$ the $l$-th successful $SC$ instruction on base object $S$.
Let $it_i$ be any iteration of the for loop of line~\ref{alg:lsimopt:attempt_loop} 
that is executed by a thread $p_i$. 
Let $SV_r(it_i)$ be the sequence of base objects read by the \LL\ instructions
of line~\ref{alg:lsimopt:ll_dir} in $it_i$. 
Denote by $|SV_r(it_i)|$ the number of elements of $SV_r(it_i)$.

For each $1 \leq j \leq |SV_r(it_i)|$,
denote by $SV_r^j(it_i)$ the prefix of $SV_r(it_i)$ containing the $j$
first elements of $SV_r(it_i)$,
i.e. $SV_r^j(it_i) = \langle sv_r^1(it_i), \dots, sv_r^j(it_i)\rangle$,
where $sv_r^j(it_i)$ is the $j$-th \LL\ instruction performed by $it_i$
%on some base object $r$.
on any base object.
Let $SV_r^0(it_i) = \lambda$ be the empty sequence.

Let $V_r(it_i)$ be the sequence of insertions in directory $D$
(lines~\ref{alg:lsimopt:add_dir1}-\ref{alg:lsimopt:add_dir2}) by $it_i$.
Denote by $|V_r(it_i)|$ the number of elements of $V_r(it_i)$.
Obviously, it holds that $|SV_r(it_i)| = |V_r(it_i)|$.
For each $1 \leq j \leq |V_r(it_i)|$,
denote by $v_r^i(it_i)$ the prefix of $V_r(it_i)$ containing the $j$
first elements of $V_r(it_i)$,
i.e. $V_r^j(it_i) = \langle v_r^1(it_i), \dots, v_r^j(it_i)\rangle$,
where $v_j(it_i)$ is the $j$-th value inserted to directory $D$.
Let $V_r^0(it_i) = \lambda$ be the empty sequence.

Let $SV_w(it_i)$ be the sequence of shared base objects 
accessed by $it_i$ while executing lines~\ref{alg:lsimopt:sc_dir1}-\ref{alg:lsimopt:sc_dir2}
(we sometimes abuse notation and say that a code line is executed by $it_i$
to denote that the code line is executed by $p_i$ during the execution of $p_i$).
Denote by $|SV_w(it_i)|$ the number of elements of $SV_w(it_i)$.
For each $1 \leq j \leq |SV_w(it_i)|$,
denote by $SV_w^j(it_i)$ the prefix of $SV_w(it_i)$ that contains the $j$
last elements of $SV_w(it_i)$, i.e. $SV_w^j(it_i) = \langle svw_1(it_i), \dots, svw_j(it_i)\rangle$,
where $svw_j(it_i)$ is the $j$-th request 
(lines~\ref{alg:lsimopt:sc_dir1}-\ref{alg:lsimopt:sc_dir2}) by $it_i$.
Let $SV_w^0(it_i) = \lambda$ be the empty sequence.

Let $SV_a(it_i)$ be the sequence of shared base objects allocations during
$it_i$ iteration (lines~\ref{alg:lsimopt:alloc_var}-\ref{alg:lsimopt:new_var}).
Denote by $|SV_a(it_i)|$ the number of elements of $SV_a(it_i)$.
For each $1 \leq j \leq |SV_a(it_i)|$,
denote by $SV_a^j(it_i)$ the prefix of $SV_a(it_i)$ that contains the $j$
first elements of $SV_a(it_i)$, i.e. $SV_a^j(it_i) = \langle sva_1(it_i), \dots, sva_j(it_i)\rangle$, where $sva_j(it_i)$ is the $j$-th base object allocation 
by $it_i$.

Let $SV_{arw}(it_i)$ be the sequence of allocations/reads/writes  
that $it_i$ performs on base objects
in lines~\ref{alg:lsimopt:alloc_var}-\ref{alg:lsimopt:sc_dir3} 
of the pseudocode. Denote by $|SV_{arw}(it_i)|$ the number of elements of $SV_{arw}(it_i)$.
Obviously, it holds that
$|SV_{arw}(it_i)| = |SV_{a}(it_i)| + |SV_{r}(it_i)| + |SV_{w}(it_i)|$.
For each $1 \leq j \leq |SV_{arw}(it_i)|$, denote by $SV_{arw}^j(it_i)$ 
the prefix of $SV_{arw}(it_i)$ that contains the $j$ first elements of sequence $SV_{arw}(it_i)$
(i.e. $SV_{arw}^j(it_i) = \langle svarw_1(it_i), \dots,$ $svarw_j(it_i)\rangle$)
where $svarw_j(it_i)$ is the $j$-th base object allocations/reads/writes of base objects
performed by $it_i$.

The next lemma states that for any process $p_i$ 
that has a pending request, the $i$-th element 
of the $Announce$ array stores the pending request of $p_i$
for an appropriate time interval.

\begin{lemma}
\label{lemma:L_announce_unmodified}
Let $l>0$ be any integer such that $S.applied[i] \neq S.papplied[i]$ 
at configuration $C_{l-1}$. Let $req_j^i$ be the value of $Announce[i]$ at $C_{l-1}$.
In any configuration between $C_{l-1}$ and $C_l$, it holds that $Announce[i] = req_j^i$.
\end{lemma}

\remove{
\begin{proof}
Assume, by the way of contradiction, that there is
at least one configuration between $C_{l-1}$ and $C_l$, such that
$Announce[i]=req_{j'}^i \neq req_j^i$. Let $C$ be the first of these
configurations. The pseudocode (line~\ref{alg:lsimopt:announce_op}) implies
that $p_i$ is the only thread that modifies base object $Announce[i]$. 
Thus, $p_i$ starts the execution of a new request $req_{j'}^i$ at $C$,
and it holds that $j' = j + 1$.
Since the write on $Announce[i]$ by $p_i$ is executed between $C_{l-1}$ and $C_l$, 
it is implied that either $C_{j+1}^i = C_l$ or $C_{j+1}^i$ follows $C_l$.
Since the end of $req_j^i$ precedes $C$, it follows that either 
$\tilde C_j^i = C_{l-1}$ or $\tilde C_j^i$ precedes $C_{l-1}$.
Lemma~\ref{lsimopt:lemma:relation_applied_papplied} implies that 
$S.applied[i] = S.papplied$ in any configuration between  $\tilde C_j^i$ 
and $C_{j+1}^i$ ($C_{j+1}^i$ is excluded).
Thus, it holds that $S.applied[i] = S.papplied$ in any configuration
between $C_{l-1}$ and $C_l$, which is contracts our claim that 
$S.applied[i] \neq S.papplied[i]$ at $C_{l-1}$.
\end{proof}
}

%%%%%%%%%%%%%%%%%%%%%%%%%%%%%%%%%%%%%%%%%%%%%%%%%%%%%%%%%%%%%%%%%%%%%%%
%%%%%%%%%%%%%%%%%%%%%%%%% Big Lemma Proof %%%%%%%%%%%%%%%%%%%%%%%%%%%%%
%%%%%%%%%%%%%%%%%%%%%%%%%%%%%%%%%%%%%%%%%%%%%%%%%%%%%%%%%%%%%%%%%%%%%%%

\begin{lemma}
\label{lsimopt:lemma:big_lemma}
Let $r$ be any shared base object other than $S$. For any $l > 0$, the following claims are true:
\begin{enumerate}
  \item At most one successful \SC\ instruction is executed on $r$ between
        $C_{l-1}$ and $C_l$. 
  \item In case that a successful \SC\ instruction $SC_w$ is executed on $r$,
        it holds that $r.seq < l$ just before $SC_w$ and $r.seq = l$ just after $SC_w$.
  \item Let $it_i$ be some iteration of the loop of line~\ref{alg:lsimopt:attempt_loop}
        executed by a thread $p_i$ that executes at least one successful
        \SC\ instruction $SC_r$ on $r$. If $LL_r$
        is the \LL\ instruction of line~\ref{alg:lsimopt:ll_iteration}
        executed by $it_i$, then $LL_r$ is executed after $C_{l-1}$.
  \item Let $it_i$, $it_{i'}$ be two iterations of the for loop of
        line~\ref{alg:lsimopt:attempt_loop} executed by threads $p_i$ and $p_{i'}$
        respectively, such that that both $it_i$, $it_{i'}$ execute their \LL\
        instructions of line~\ref{alg:lsimopt:ll_iteration} somewhere between
        $C_{l-1}$ and $C_l$, $l > 0$, and $|SV_{arw}(it_i)| \geq |SV_{arw}(it_{i'})|$. 
        If both $it_i$, $it_{i'}$
        execute line~\ref{alg:lsimopt:flush_dir}, just before $C_l$ it holds 
        that $SV_{arw}(it_i) = SV_{arw}(it_{i'})$.
        %that:
        %(i)  $SV_r(it_i) = SV_r(it_{i'})$ and $V_r(it_i) = V_r(it_{i'})$, 
        %(ii) $SV_a(it_i) = SV_a(it_{i'})$, and
        %(iii) $SV_w(it_i) = SV_w(it_{i'})$.
\end{enumerate}
\end{lemma}

\begin{proof}
We prove the claims by induction on $l$.
Fix any $l \geq 1$ and assume that the claims hold for $l$.
We prove that the claims hold for $l+1$.

We first prove Claim 1. Let $SC'$ be the first of
the successful \SC\ instruction on $r$ between $C_{l-1}$ and $C_l$.
We prove that $r.seq = l$ just after the execution of $SC'$. Assume by the way
of contradiction that $r.seq = l' \neq j$. Let $it_h$ be the iteration of
line~\ref{alg:lsimopt:ll_iteration} executed by some thread $p_h$ that
executes $SC'$. Let $LL'$ be the matching \LL\ instruction of $SC'$.
Since $it_i$ executes successfully line~\ref{alg:lsimopt:sc_dir2} of the pseudocode,
the pseudocode (lines~\ref{alg:lsimopt:vl} and~\ref{alg:lsimopt:sc_dir2}) implies
that the \VL\ instruction of line~\ref{alg:lsimopt:vl} returns \TRUE.
Since $LL'$ is executed by $it_i$ before this \VL\ instruction,
it follows that $LL'$ precedes $SC_{j'}$. Thus, the \VL\ instruction
of line~\ref{alg:lsimopt:vl} is executed before $SC_{j'}$.
Let $it_{i'}$ be the iteration of the loop of line~\ref{alg:lsimopt:ll_iteration}
at which $SC_{j'}$ is executed and let $p_{i'}$ be the thread that
executes $SC_{j'}$. Obviously, $LL_{j'}$ has been executed between
$C_{l'-1}$ and $C_{l'}$. Since $LL'$ is also executed between $C_{l'-1}$
and $C_{l'}$, the induction hypothesis (Claim $2$.ii) implies that
$SV_w(it_h) = SV_w(it_q)$.
Thus, $it_q$ has also executed an \SC\ instruction on $r$.
By lines \ref{alg:lsimopt:ll_dir}, \ref{alg:lsimopt:flush_dir}-\ref{alg:lsimopt:sc_dir2} and 
\ref{alg:lsimopt:sc_on_s} of the pseudocode, it follows that
there is a successful \SC\ instruction on $r$ between $SC_{l'-1}$ and
$SC_{l'}$. Let $SC_r$ be this instruction. By induction hypothesis (claim 1),
it follows that $r.seq=j'$ just after the execution of $SC_r$. Since $SC'$
is a successful \SC\ instruction, $LL'$ follows $SC_r$. By the pseudocode
(lines~\ref{alg:lsimopt:sc_dir1}-\ref{alg:lsimopt:sc_dir2}),
it follows that $SC'$ is not executed, which is a contradiction.
Therefore $r.seq = j$ just after the execution of $SC_r$. We now prove
that there is no other successful \SC\ instruction between $SC'$ and $C_l$ on $r$.
Assume by the way of contradiction that at least one successful \SC\
instruction takes place between $SC'$ and $C_l$. Let $SC''$ be the first
of these instructions. Since, $SC''$ is a successful \SC\ instruction, it
follows that its matching \LL\ instruction $LL''$ follows $SC'$. By the
pseudocode (lines~\ref{alg:lsimopt:sc_dir1}-\ref{alg:lsimopt:sc_dir2}),
it follows that $SC''$ is not executed since $r.seq = S.seq$,
which is a contradiction.

Claim~2 is proved using a similar argument as that above for Claim~1.

We now prove Claim~3. Assume by the way of contradiction that $LL_p$ is executed
between $SC_{j'-1}$ and $SC_{j'}$, $j'<j$. Let $p_i$ be the thread that executes $SC_{j'}$ 
on some iteration $it_i$. By Claim~1 and by Claim~2, it follows that
$r.seq \leq j'$ just before $SC_{j'}$. Thus $SC_r$ is not executed, which is a contradiction.
Thus, Claim~3 holds.

To prove Claim 4, it is enough to prove that $svarw_{l'}(it_i) = svarw_{l'}(it_{i'})$,
for any $l' \leq |SV_{arw}(it_i)|$.
We prove this claim by induction on the number $l' \leq |SV_{arw}|$ of 
elements of $SV_{arw}(it_i)$ (see appendix).
\end{proof}

%%%%%%%%%%%%%%%%%%%%%%%%%%%%%%%%%%%%%%%%%%%%%%%%%%%%%%%%%%%%%%%%%%%%%%%
%%%%%%%%%%%%%%%%%%%%%%%% Consistency Proof %%%%%%%%%%%%%%%%%%%%%%%%%%%%
%%%%%%%%%%%%%%%%%%%%%%%%%%%%%%%%%%%%%%%%%%%%%%%%%%%%%%%%%%%%%%%%%%%%%%%

Denote by $\alpha_i$, the prefix of $\alpha$ which ends at $SC_i$ and let $C_i$
be the first configuration following $SC_i$. Let $\alpha_0$ be the empty
execution. Denote by $l_i$ the linearization order of the requests in $\alpha_i$.

We are now ready to prove that $a_i$ is linearizable. This require to prove
that the object state is consistent after the execution of each successful
\SC\ on $S$.

\begin{lemma}
	\label{lem:linearizability}
For each $i \geq 0$, the following claims hold:
	\begin{enumerate}
		\item object's state is consistent at $C_i$, and 
		\item $\alpha_i$ is linearizable.
	\end{enumerate}
\end{lemma}

\begin{proof}
We prove the claim by induction on $i$.
The claim holds trivially; we remark that $\alpha_i$ is empty in this case.
Fix any $i>0$ and assume that the claim holds for  $i-1$.
We prove that the claim holds for $i$.

By the induction hypothesis,
it holds that: (1) object's state is consistent at $C_{i-1}$,
and (2) $\alpha_{i-1}$ is consistent with linearization $l_{i-1}$.
Let $req$ be the request that executes $SC_i$. If $req$ applies no request
on the simulated object, the claim holds by induction hypothesis. Thus,
assume that $req$ applies $j > 0$ requests on the simulated object.
Denote by $req_1, ..., req_j$ the sequence of
these requests ordered with respect to the identifiers of the threads
that initiate them.

Notice that $req$ performs $LL_i$ after $C_{i-1}$ since otherwise
$SC_i$ would not be successful.
By the induction hypothesis, object's is consistent at $C_{i-1}$.
By Lemma \ref{lsimopt:lemma:relation_applied_papplied}, Observation~\ref{lsimopt:obs:relation_applied_papplied},
Lemma~\ref{lsimopt:lemma:req_applied_at_least_once}, and of the definition of $\tilde C_j^i$,
it follows that 
each request $req$ is applied exactly once.
Thus, Lemma~\ref{lsimopt:lemma:big_lemma} 
imply that all threads that are trying to apply a set of requests
between $C_{i-1}$ and $C_{i}$ do the following (1) apply the same set of requests
with the same order, (2) all read the same consistent state of the object,
(3) write the same set of base objects with the same values
(although only one write succeeds), and (4) none of $req_1, \ldots, req_j$
have been applied in the past.

Given that $req_1, ..., req_j$ are executed by $req$ sequentially, the one after
the other in the order mentioned above,
it is a straightforward induction to prove that (1) for each $f$, $0 \leq f \leq j$, request $req_f$
returns a consistent response;
% Edw exoume ena themataki
% den einai $S \rightarrow st$
% alla $Pool[x][y] \rightarrow st$
% alla pws to leme auto....?
moreover, $S \rightarrow st$ is consistent and once line $14$ has been executed by $req$ for all these requests.
Therefore,  $S \rightarrow st$ is consistent after the execution of $req$'s successful \SC.
This concludes the proof of the claim.
\end{proof}

Lemma~\ref{lem:linearizability} implies that \LSIMOPT\ is linearizable. 
The discussion in Section~\ref{step complexity} implies that
\LSIMOPT\ is also wait-free and its step complexity is $O(n+kw)$.
Thus:

\begin{theorem}
\label{theorem:lsimopt}
\LSIMOPT\ is a linearizable, wait-free implementation of a universal object.
The number of shared memory accesses performed by \LSIMOPT\ is $O(n+kw)$.
\end{theorem}

\end{document}